\documentclass[a4paper,10pt]{llncs}

\usepackage{graphicx}
\usepackage{amsmath,amssymb}
\usepackage{braket}
\usepackage{algorithm}
\usepackage{algorithmic}
\usepackage{color}

\usepackage{url}
\urldef{\mailsa}\path|armknecht@uni-mannheim.de|
\urldef{\mailsb}\path|tommaso.gagliardoni@cased.de|
\urldef{\mailsc}\path|katzenbeisser@seceng.informatik.tu-darmstadt.de|
\urldef{\mailsd}\path|a.peter@utwente.nl|
\newcommand{\keywords}[1]{\par\addvspace\baselineskip
\noindent\keywordname\enspace\ignorespaces#1}

\newcommand{\group}{G}
\newcommand{\subgroup}{H}
\newcommand{\nsgroup}{\subgroup}
\newcommand{\neutr}{1}
\newcommand{\op}{\cdot}
\newcommand{\inv}[1]{{#1}^{-1}}

\newcommand{\adv}{\mathcal A}
\newcommand{\bdv}{\mathcal B}
\newcommand{\negl}{{\tt negl}}
\newcommand{\prob}[2]{\underset{#1}{\rm Pr}\left[#2\right]}

\newcommand{\compind}{\overset{c}{=}}
\newcommand{\fromdist}{\longleftarrow}
\newcommand{\fromunif}{\overset{U}{\longleftarrow}}

\newcommand{\secparam}{\lambda}
\newcommand{\advantage}[2]{{\rm Adv}_{#1}^{#2}}
\newcommand{\indcpa}{{\rm IND-CPA}}
\newcommand{\encsch}{\mathcal E}
\newcommand{\keygen}{{\sf KeyGen}}
\newcommand{\enc}{{\sf Enc}}
\newcommand{\dec}{{\sf Dec}}
\newcommand{\pk}{{\sf pk}}
\newcommand{\sk}{{\sf sk}}
\newcommand{\ptexts}{\mathcal P}
\newcommand{\ctexts}{\widehat{\mathcal C}}
\newcommand{\allencs}{\mathcal C}
\newcommand{\Rnd}{{\sf Rnd}}
\newcommand{\indexp}[2]{{\bf Exp}^{{\rm ind}\text{-}#1}_{#2}}
\newcommand{\cpa}{{\rm cpa}}
\newcommand{\smp}{{\rm smp}}
\newcommand{\smps}{{\rm smp}^*}
\newcommand{\cipher}[1]{\allencs_{#1}}
\newcommand{\SMP}{{\rm SMP}}

\newcommand{\gen}{{\sf Gen}}

\newcommand{\nepsilon}{{\tt \varepsilon}}

\newcommand{\qadv}{\mathcal A_Q}
\newcommand{\genby}[1]{\left\langle #1\right\rangle}
\newcommand{\poly}{{\tt poly}}

\newcommand{\dist}{\mathcal D}
\newcommand{\nnegl}{\Gamma}

\newcommand{\GHE}{group homomorphic encryption}

\newcommand{\eval}{{\sf Eval}}
\newcommand{\grp}[1]{\group_{0,1}}
\newcommand{\cover}{\delta} 
\newcommand{\coversuccess}{\sigma} 
\newcommand{\coversamples}{N} 
\newcommand{\Sample}{\mathsf{Sample}}

\newcommand{\ord}{\mathrm{ord}}

\renewcommand{\qed}{{ }\hfill$\Box$}
\renewcommand{\subset}{\subseteq}

\renewcommand{\output}{\fromdist}

\renewcommand{\exp}[2]{{\bf Exp}^{#1}_{#2}}
\renewcommand{\phi}{\varphi}

\newcommand{\fromdistr}[1]{\overset{#1}{\longleftarrow}}
\newcommand{\tommaso}[1]{ }

\newcommand{\loosebound}{\left\lceil \frac{\log(1-\coversuccess) - \log(k)}{\log(\cover)}\right\rceil}

\newtheorem{fact}{Fact}
\newtheorem{observation}{Observation}

\newenvironment{smpexp}{\vspace{3mm}\newline{\rm Experiment} $\exp{\smp}{\adv,\gen}(\secparam)${\rm :}
\begin{enumerate}}{\end{enumerate}\par\noindent\ignorespacesafterend}

\begin{document}

\mainmatter

\title{General Impossibility of Group Homomorphic Encryption in the Quantum World}

\author{Frederik Armknecht\inst{1}
\and Tommaso Gagliardoni\inst{2,}\thanks{Supported by the German Federal Ministry of Education and Research (BMBF) within the EC-SPRIDE project.}
\and \\ Stefan Katzenbeisser\inst{2}
\and Andreas Peter\inst{3,}\thanks{Supported by the THeCS project as part of the Dutch national program COMMIT.}}
%
\institute{Universit\"at Mannheim, Germany\\
\mailsa
\and Technische Universit\"at Darmstadt and CASED, Germany \\
\mailsb \\
\mailsc
\and University of Twente, The Netherlands \\
\mailsd}

\maketitle

\begin{abstract}
Group homomorphic encryption represents one of the most important building blocks in modern cryptography. It forms the basis of widely-used, more sophisticated primitives, such as CCA2-secure encryption or secure multiparty computation. Unfortunately, recent advances in quantum computation show that many of the existing schemes completely break down once quantum computers reach maturity (mainly due to Shor's algorithm). This leads to the challenge of constructing quantum-resistant group homomorphic cryptosystems.

In this work, we prove the \emph{general} impossibility of (abelian) group homomorphic encryption in the presence of quantum adversaries, when assuming the \indcpa\ security notion as the minimal security requirement. To this end, we prove a new result on the probability of sampling generating sets of finite (sub-)groups if sampling is done with respect to an arbitrary, unknown distribution.  Finally, we provide a sufficient condition on homomorphic encryption schemes for our quantum attack to work and discuss its satisfiability in non-group homomorphic cases. The impact of our results on recent fully homomorphic encryption schemes poses itself as an open question.
\end{abstract}



\section{Introduction}\label{sec:introduction}
Since the introduction of public-key cryptography by Diffie and
Hellman~\cite{diffie76pkcrypto} in 1976, researchers strived to
construct encryption schemes that are \emph{group homomorphic}. This
property can be characterized by requiring the encryption scheme to
have a homomorphic decryption procedure, while the plaintext and
ciphertext spaces form groups. Ever since, the topic of homomorphic
encryption is of central importance in cryptography. The recent
advances in fully homomorphic encryption
(FHE)~\cite{brakerski12fhe,gentry09fhe,gentry12aes} constitute just
one example of this trend.
In practice, \emph{group} homomorphic encryption schemes lie at the
heart of several important applications, such as electronic
voting~\cite{CraGS97}, private information retrieval~\cite{KusO97},
or multiparty computation~\cite{CraDN01} to name just a few.
Moreover, the group homomorphic property comes quite naturally, as
witnessed by a number of encryption schemes, for example
RSA~\cite{rivest78rsa}, ElGamal~\cite{elgamal85elgamal},
Goldwasser-Micali~\cite{goldwasser84probenc}, where the homomorphic property was not a design goal, but rather arose ``by~chance''.

So far, these cryptosystems were all analyzed in the classical model
of computation. However, it is reasonable to assume that the \emph{quantum}
model of computation will become more realistic in the
future. Unfortunately, in this model all aforementioned cryptosystems
are insecure due to Shor's algorithm~\cite{shor94quantum}, which
allows to efficiently solve the discrete logarithm problem and to
factor large integers. That is, until today nobody has been able to come up with a group homomorphic encryption scheme that can
withstand quantum attackers. 

It seems that such a scheme would require other design approaches. For
instance, when considering ElGamal-like encryption schemes, simply
replacing the underlying computational hardness assumption by a
supposedly quantum-resistant one, say code-based, is not
enough~\cite{armknecht10soap}. 
In fact although there is a substantial number
of classical cryptographic primitives that can be proven secure
against quantum attackers, e.g.~\cite{hallgren11quantum}, we still
know little about what classical primitives can be realized in the
quantum world and what not. Indeed this applies to the case of group homomorphic encryption schemes as well: so far it was even undecided whether \emph{group
homomorphic encryption can exist at all in the quantum world}. In other words, does the absence of a quantum secure group homomorphic encryption scheme so far imply that the right approach has not been found yet (but may be in the future) or are there universal reasons that prevent the existence of such schemes?

\subsection{Our Contributions}
\subsubsection{Basic Impossibility Result.}
The central contribution of this work is to give a \textit{negative}
answer to the above question: 
\begin{center}
\it
It is impossible to construct secure group homomorphic encryption in the quantum
world, if the plaintext and ciphertext spaces form abelian groups.
\end{center}
More precisely, we prove that any such scheme\footnote{Although we postulate that our result is extendible to arbitrary solvable groups, we focus on the abelian case, since it is the most important one for reasons of practicability in real-world applications.} cannot meet the minimial security notion of \indcpa\ security in the presence of quantum adversaries. Observe that this result not only re-confirms the insecurity of \emph{existing} schemes, but shows that \emph{all} group homomorphic encryption schemes (including all yet to come schemes) are inevitably insecure in the quantum world.
\subsubsection{Quantum Attack.} In order to prove this impossibility, we start by exhibiting the fact that the \indcpa\ security of any group homomorphic encryption scheme can be reduced to an abstract \emph{Subgroup Membership Problem} (\SMP), introduced by Cramer and Shoup~\cite{cramer02uhp}, which is much easier to analyze. Roughly speaking, this problem states that given a group $G$ with subgroup $H$ and a randomly sampled (according to some arbitrary distribution $\dist$) element $g\in G$, decide whether $g\in H$ or not.
This reduction to the \SMP\ tells us that in order to break the \indcpa\ security of a given group homomorphic encryption scheme in the quantum world, it is sufficient to give a quantum algorithm that breaks the \SMP.
Now, the basic idea for breaking the \SMP\ for groups $(G,H)$ is to use Watrous' variant~\cite{watrous01qalggrp} of the famous group order-finding quantum algorithm, which will effectively decide membership.
\vspace{-0.2cm}
\subsubsection{Sampling Generators in Finite Groups.}
Unfortunately, this algorithm only works when given a set of generators of $H$ which we commonly do not have. Hence we restrict to the generic case that an attacker has only access to an efficient \emph{sampling algorithm} for $H$ that samples according to some distribution $\dist$. We distinguish between the following two cases:
\begin{itemize}
\item  \textbf{Uniform Distribution.}
If $\dist$ is uniform, Erd\"os and R\'enyi~\cite{erdos65sampling} show that sampling polynomially many times from $H$ will give a generating set with high probability---a result that has been improved by Pak and Bratus~\cite{pak99sampling}: If $k=\lceil\log_2(|H|)\rceil$, then $k+4$ samples are enough to get a set of generators with probability $\geq 3/4$.
After obtaining a generating set for $H$, we use Watrous' quantum algorithm to decide membership in $H$, and hence efficiently break the \SMP\ for $(G,H)$.
\item \textbf{Arbitrary/Unknown Distribution.}
In general, the distribution $\dist$ does not have to be uniform, but can be arbitrary, or completely unknown. Interestingly, we prove that, even then, breaking the SMP is possible with (almost) linearly many samples only.
Observe that as we do not make any restrictions on the sampling algorithm, we  cannot exclude seemingly exotic cases where regions of $\nsgroup$ are hardly (or never) reached by the sampling algorithm. Thus, the best we can aim for is to find a generating set for a subgroup $\nsgroup^*$ of $\nsgroup$ such that the probability that a random sample (with respect to $\dist$) does fall into $\nsgroup^*$ is above an arbitrarily chosen threshold $\delta$. 
We call such subgroups to be {\em $\delta$-covering}.  It turns out that having a generating set for such a subgroup is enough to break the \SMP\ for $(G,H)$.
The main challenge, however, is to find a generating set for a $\delta$-covering subgroup. To this end, we prove a new result on the probability of sampling generating sets of finite (sub-)groups with unknown sampling distribution. More precisely, we show that for any chosen probability threshold $\delta^*$, there exists a value $N$, which grows at most logarithmically in $k$ and does not depend on $\dist$, such that $N\cdot k+1$ samples yield a generating set for a $\delta$-covering subgroup with probability at least $\delta^*$. This result represents one of the main technical contribution of our work. We believe that it is also applicable in other research areas, e.g., computational group theory, and hence might be of independent interest.
\end{itemize}
\subsubsection{Possible Extensions to Fully Homomorphic Encryption Schemes.}
Finally, we provide a general sufficient condition on a homomorphic encryption scheme for our quantum attack to work and discuss the applicability in FHE schemes. The decision of whether our attack breaks any of the existing FHE schemes~\cite{brakerski12fhe,gentry09fhe,gentry12aes} proves itself to be a highly non-trivial task and lies outside the scope of this paper. We leave it as interesting future work.
\vspace{-0.1cm}
\subsection{Related Work}
There are many papers dealing with the construction of \indcpa\ secure group homomorphic encryption schemes~\cite{paillier99grphom,gjosteen05homsmp,damgaard08grphom,armknecht10soap,peter12grphom}. Some of these works attempt\-ed to build such schemes using post-quantum primitives~\cite{armknecht11coding}, which did not succeed (for a good reason as our results show). Also, for a restricted class of group homomorphic schemes, \cite{armknecht10soap} shows the impossibility of using linear codes as the ciphertext group. Furthermore, we mention the impossibility (even in the classical world) of \emph{algebraically homomorphic} encryption schemes~\cite{boneh96blackbox}, which are deterministic encryption schemes and thus do not fall into the class \indcpa\ secure cryptosystems.

In the quantum world, there is an even more efficient algorithm for breaking such algebraically homomorphic schemes~\cite{dam06quantum}. In this vein, there are many variants of Shor's algorithm~\cite{shor94quantum} that are being used to solve different computational problems~\cite{mosca12quantum,watrous01qalggrp}, leading to the breakdown of certain cryptosystems. On the other hand, there are several papers dealing with the analysis of classical primitives in the presence of quantum adversaries~\cite{hallgren08zero,hallgren11quantum}. However, none of these works show a general impossibility of group homomorphic cryptosystems.

With respect to the sampling from finite groups, there are many papers that are concerned with the improvement of probability bounds on finding generating sets when sampling uniformly at random~\cite{erdos65sampling,babai91sampling,pak99sampling}. Similar strong results for the arbitrary sampling from finite groups are not known.

Finally, we mention the recent advances in fully homomorphic encryption (FHE)~\cite{brakerski12fhe,gentry09fhe,gentry12aes}. These schemes are not classified as being group homomorphic, as they follow a different design approach. Rather than having a group homomorphic decryption algorithm, the decryption is only guaranteed to run correctly for polynomially many evaluations of the group operation.
Interestingly enough, our results show that since current FHE schemes are based on post-quantum hardness assumptions, they had to follow a different approach than the group homomorphic one.
\vspace{-0.1cm}
\subsection{Outline}
We recall standard notation in Section~\ref{sec:notation} and show some basic observations on group homomorphic encryption and the Subgroup Membership Problem (\SMP) in Section~\ref{sec:grphom}. Section~\ref{sec:quantum} covers the main Theorem, showing the impossibility of group homomorphic encryption in the quantum world, thereby giving our new insights in the sampling of group generators. We discuss non-group homomorphic encryption, such as somewhat and (leveled) fully homomorphic encryption in Section~\ref{sec:discussion}.


\section{Notation}\label{sec:notation}
Throughout the paper, we use some standard notation that we briefly want to recall. We write $x\fromdist X$ if $X$ is a random variable or distribution and $x$ is to be chosen randomly from $X$ according to its distribution. In the case where $X$ is solely a set, $x\fromunif X$ denotes that $x$ is chosen uniformly at random from $X$. If we sample an element $x$ from $X$ by using a specific distribution $\dist$, we write $x\fromdistr{\dist}X$ (or $x\output X$ when there is no doubt about the distribution $\dist$). For a distribution $\dist$ on $X$, the term $\dist(x)$ for $x\in X$ expresses the probability with which $x$ is sampled according to $\dist$, i.e., the probability mass function at $x\in X$.

For an algorithm $\adv$ we write $x\output\adv(y)$ if $\adv$ outputs $x$ on fixed input $y$ according to $\adv$'s distribution. Sometimes, we need to specify the randomness of a probabilistic algorithm $\adv$ explicitly. To this end, we interpret $\adv$ in the usual way as a deterministic algorithm $\adv(y;r)$, which has access to values $r\output\Rnd$ that are randomly chosen from some randomness space $\Rnd$. Moreover, two distribution ensembles $X=\{X_\secparam\}_{\secparam\in\mathbb N}$ and $Y=\{Y_\secparam\}_{\secparam\in\mathbb N}$ taking values in a finite set $S_\secparam$ (indexed by a parameter $\secparam$) are said to be \emph{computationally indistinguishable}, if for all probabilistic polynomial time (PPT) algorithms $\adv$ there exists a negligible function $\negl$ such that
\[\advantage{\adv}{X,Y}(\secparam):=\left|\prob{x\output X_\secparam}{\adv(x)=1}-\prob{y\output Y_\secparam}{\adv(y)=1}\right|\leq\negl(\secparam).\]
We denote this by $X\compind Y$. 

For a group $\group$, we denote the neutral element by $\neutr$, and denote the binary operation on $\group$ by ``$\op$'', i.e., $\group$ is written in \emph{multiplicative notation}. We recall that a subgroup $\subgroup$ of a group $\group$ is said to be \emph{normal} if $z\op h\op \inv{z}\in\subgroup$ for all $z\in\group, h\in\subgroup$. In particular, this means that if $\group$ is an abelian group, then every subgroup $\subgroup$ is normal.

In general, we will consider sequences of abelian groups $(\group_\secparam)_\secparam$ indexed by a parameter $\secparam$, where any element of every $\group_\secparam$ admits a representation of size at most polynomial in $\secparam$. We might assume, without loss of generality, that the choice of this polynomial is the identity, and in particular that every $\group_\secparam$ has order upper bounded by $2^\secparam$. We will just write $\group$ instead of $\group_\secparam$ for any fixed choice of $\secparam$.

By a \emph{description} of a finite group $\group$ we mean an efficient (i.e., PPT in $\secparam$) sampling algorithm (where sampling is denoted by $x\output\group$), the neutral element $\neutr$, an efficient algorithm for performing the group operation on $\group$, and one for the inversion of group elements. Notice that the output distribution of the sampling algorithm does not have to be necessarily uniform. We abuse notation and write $\group$ both for the description and for the group itself. Furthermore, for elements $x_1,\ldots,x_k\in\group$, we write $\genby{x_1,\ldots,x_k}$ for the subgroup generated by $x_1,\ldots,x_k$.

\section{Group Homomorphic Encryption}\label{sec:grphom}
We recall the notion of public-key \emph{group homomorphic} encryption, which roughly can be described as usual public-key encryption where the decryption algorithm is a group homomorphism. 
\begin{definition}[Group Homomorphic Encryption~\cite{armknecht10soap,hemenway12homcca}]\label{def:grouphom}
A public key encryption scheme $\encsch=(\keygen,\enc,\dec)$ is called \emph{group homomorphic}, if for every output $(\pk,\sk)$ of $\keygen(\secparam)$, the plaintext space $\ptexts$ and the ciphertext space $\ctexts$ are non-trivial groups such that
\begin{itemize}
\item the \emph{set of all encryptions} $\allencs:=\{\enc_\pk(m;r)\mid m\in\ptexts, r\in\Rnd\}$ is a non-trivial subgroup of $\ctexts$
\item the decryption $\dec_\sk$ is a group homomorphism on $\allencs$, i.e.
\[\dec_\sk(c\op c')=\dec_\sk(c)\op\dec_\sk(c'),\text{ for all }c,c'\in\allencs.\footnote{Note that the decryption might output an error $\bot$ on inputs in $\ctexts\setminus\allencs$. Therefore, requiring it to be homomorphic on $\allencs$ is as general as possible since we do not give any restriction on its behaviour outside of $\allencs$.}\]
\end{itemize}
\end{definition}
Notice that the scheme does {\em not} include a membership testing algorithm (i.e., an algorithm to test whether a group element is a valid encryption or not). The standard security notion for such homomorphic encryption schemes is that of \emph{indistinguishability under chosen-plaintext attack}, denoted by \indcpa~\cite{armknecht10soap}. Informally, this notion states whenever an adversary picks two plaintext messages of his choosing and gets to see an encryption of either of them, it should be computationally infeasible for him to decide which of the two messages was encrypted. Formally, for a given security parameter $\secparam$, group homomorphic encryption scheme $\encsch=(\keygen,\enc,\dec)$, and PPT adversary $\adv$, we consider the experiment $\indexp{\cpa}{\adv,\keygen}(\secparam)$, where $\adv$ chooses two different plaintexts $m_0,m_1$ and is then provided an encryption $\enc_\pk(m_b)$ for a randomly chosen bit $b$ and a public key $\pk$ output by $\keygen(\secparam)$. The experiment succeeds (outputs $1$) if $b$ is guessed correctly.
We say that $\encsch$ is \indcpa\ \emph{secure} if the advantage
\[\left|\prob{}{\indexp{\cpa}{\adv,\keygen}(\secparam)=1}-\frac{1}{2}\right|\text{ is negligible for all PPT adversaries $\adv$.}\]
Moreover, we recall a fact showing the strong group-theoretic structure of the set of encryptions of $\neutr\in\ptexts$ for \emph{any} group-homomorphic encryption scheme. For this, we introduce the \emph{set of all encryptions of $m\in\ptexts$}
\[\cipher{m}:=\{c\in\allencs\mid\dec_\sk(c)=m\}.\]
\begin{fact}[Basic Properties~\cite{armknecht10soap}]\label{lemma:c0}
Let $\encsch=(\keygen,\enc,\dec)$ be an arbitrary group homomorphic encryption scheme. It holds that 
\begin{enumerate}
\item $\cipher{m}=\enc_\pk(m;r)\op\cipher{\neutr}$ for all $m\in\ptexts$ and all $r\in\Rnd$, and
\item $\cipher{\neutr}$ is a proper \emph{normal} subgroup of $\allencs$ such that $|\cipher{\neutr}|=|\cipher{m}|$ for all $m\in\ptexts$.
\end{enumerate}
It follows that the set $\{\enc_\pk(m;r)\mid m\in\ptexts\}$ for a fixed $r$ is a system of representatives of $\allencs/\cipher{\neutr}$.
\end{fact}
With this notation, the \indcpa\ security of $\encsch$ is equivalent to saying that the distribution on $\cipher{m_0}$ (induced by the encryption algorithm $\enc_\pk(m)$) is computationally indistinguishable from the distribution on $\cipher{m_1}$ for any two messages $m_0$ and $m_1$~\cite[Ch.~5.2]{goldreich04foc}, i.e., $\cipher{m_0}\compind\cipher{m_1}$.

\medskip
\noindent{\bf Necessary Security Condition.}
We briefly recall the \emph{Subgroup Membership Problem} (\SMP) which was introduced by Cramer and Shoup in \cite{cramer02uhp}. 
\begin{definition}[Subgroup Membership Problem]\label{def:SMP}
Let $\gen$ be a PPT algorithm that takes a security parameter $\secparam$ as input and outputs descriptions $(\group,\nsgroup)$ where $\nsgroup$ is a non-trivial, proper subgroup of a finite group $\group$. Additionally, we assume here that there is an algorithm that allows for the efficient sampling from $\group\setminus\nsgroup$. We consider the following experiment for a given algorithm $\gen$, algorithm $\adv$ and parameter $\secparam$:
\begin{smpexp}
\item $(\group,\nsgroup)\output\gen(\secparam)$
\item Choose $b\fromunif\{0,1\}$. If $b=1$: $z\output\group\setminus\nsgroup$. Otherwise: $z\output\nsgroup$.
\item $d\output\adv(\group,\nsgroup,z)$ where $d\in\{0,1\}$
\item The output of the experiment is defined to be 1 if $d=b$ and 0 otherwise.
\end{smpexp}
We say that the \SMP\ is \emph{hard for $(\group,\nsgroup)$} (or \emph{relative to $\gen$}) if the advantage
\[\left|\prob{}{\exp{\smp}{\adv,\gen}(\secparam)=1}-\frac{1}{2}\right|\text{ is negligible for all PPT algorithms $\adv$.}\]
\end{definition}

We stress the fact that the efficient sampling from $\group\setminus\nsgroup$ does not have to be uniform. Let $\encsch=(\keygen,\enc,\dec)$ be a group homomorphic encryption scheme with the group $\allencs$ of all encryptions and the subgroup $\cipher{1}$ of all encryptions of the neutral element $1$. In fact, the hardness of \SMP\ for $(\allencs,\cipher{1})$ (i.e., relative to $\keygen$) is a necessary condition for $\encsch$ to be \indcpa\ secure. Recall that the sampling algorithms for the groups $\allencs$ and $\cipher{1}$ are the ones inherited from the encryption algorithm of $\encsch$. In particular, sampling an element $c$ from $\allencs\setminus\cipher{1}$ is done by choosing a random message $m\in\ptexts$ with $m\neq 1$ and then computing $c$ as $\enc_\pk(m;r)$ for $r\output\Rnd$. We have the following immediate result:
\begin{theorem}[Necessary Condition on IND-CPA Security]\label{thm:cpacondition}
For a group homomorphic encryption scheme $\encsch=(\keygen,\enc,\dec)$ we have:
\[\encsch\text{ is \indcpa\ secure }\Longrightarrow\text{ \SMP\ is hard (relative to $\keygen$)}.\]
\end{theorem}
The above holds regardless of the type of adversary (i.e., classical vs quantum) taken into account. A straightforward proof of this Theorem can be found in Appendix~\ref{apx:proofofcond}.
Since it is a popular belief (and for reasons of completeness), we want to point out that the converse of the Theorem does \emph{not} hold in general. This can be seen by considering a somewhat pathological example, which we present in Appendix~\ref{apx:example}.
Note that the converse of Theorem~\ref{thm:cpacondition} \emph{does}, however, hold for so-called \emph{shift-type homomorphic encryption schemes}~\cite{armknecht12shift}, which describe a certain subclass of group homomorphic encryption schemes that actually encompasses all existing instances. Furthermore, it also holds for bit encryption schemes, since there are only two messages, 0 and~1.



\section{General Impossibility in the Quantum World}\label{sec:quantum}
Let $\gen$ be a PPT algorithm that takes a security parameter $\secparam$ as input and outputs descriptions $(\group,\nsgroup)$ where $\nsgroup$ is a non-trivial, proper subgroup of a finite group $\group$ with an additional algorithm for the efficient sampling from $\group\setminus\nsgroup$ (cf.\ Section~\ref{sec:grphom}). Now, assume that for any such algorithm $\gen$, we can construct a quantum algorithm $\qadv$ that breaks the hardness of \SMP\ relative to $\gen$. In particular, for a given group homomorphic encryption scheme $\encsch=(\keygen,\enc,\dec)$ this means that we have a quantum algorithm $\qadv$ that breaks the hardness of \SMP\ relative to $\keygen$. However, by Theorem~\ref{thm:cpacondition}, this implies that we can construct an algorithm that breaks the \indcpa\ security of $\encsch$. Since we had no restriction on the encryption scheme $\encsch$, this would imply that \emph{any} group homomorphic encryption scheme $\encsch$ is insecure in terms of \indcpa\ in the quantum world. This is the result we want to prove in this section, at least for the abelian case, i.e., when $\group$ is an abelian group. Therefore, let $\gen$ be as above but with $\group$ being abelian.

It is well-known that a modification of the famous order-finding quantum algorithm~\cite{watrous01qalggrp} can efficiently find the order of an abelian group, given that we have its description by a set of generators.
\begin{theorem}[\hspace{-0.2 mm}Quantum \hspace{-0.2 mm}Order-Finding \hspace{-0.2 mm}Algorithm \hspace{-0.2 mm}with \hspace{-0.2 mm}Generators~\cite{watrous01qalggrp}]\label{thm:ordgen}
Let $\group$ be a finite abelian group with $k = \lceil \log_2(|\group|) \rceil$. Then, there exists a quantum algorithm which, given a generating set of $\group$ and an error probability $\nepsilon$ as an input, outputs the order of $\group$ with probability at least $1-\nepsilon$ in 
time $o(\poly(k+\log_2(1/\nepsilon)))$.
\end{theorem}
This Theorem already is sufficient to break the hardness of \SMP\ (relative to $\gen$), \emph{if} the description of $\nsgroup$ contains a set of generators, as the next Theorem shows.
\begin{theorem}[Quantum Attack on SMP with Generators]\label{thm:membgen}
Let $(\group,\nsgroup)$ be the output of $\gen(\secparam)$, for some security parameter $\secparam$, such that $\nsgroup$ contains a set of generators $g_1,\ldots,g_r$. Since $\gen$ is a PPT algorithm, this implies that $k=k(\secparam)=\lceil \log_2(|\nsgroup|) \rceil$ is a polynomial in $\secparam$. There exists a quantum algorithm which, given $g_1,\ldots,g_r$ (i.e., the description of $\nsgroup$), breaks the hardness of \SMP\ with probability at least $\left(1-\nepsilon\right)^2$ in time $o(\poly(k+\log_2(1/\nepsilon)))$.
\end{theorem}
\begin{proof}
Let $z$ denote the challenge in the \SMP\ game (Def.~\ref{def:SMP}), i.e., $z\in\group\setminus\nsgroup$ if $b=1$, and $z\in\nsgroup$ otherwise. Since $\nsgroup$ contains a set of generators $g_1,\ldots,g_r$, we can run the quantum algorithm in Theorem~\ref{thm:ordgen} twice: the first time on the generating 
set and the second time on the generating set plus the element $z$. Provided that both runs succeed, we have that $z\in\nsgroup$ (i.e., $b=0$) if and only if the two subgroup orders, obtained from the two algorithm runs, are the same. But both runs succeed with probability $\left(1-\nepsilon\right)^2$. This proves the Theorem.
\qed
\end{proof}
Recall that the original definition of \SMP\ gives no set of generators for $\nsgroup$ a priori, since the description of a group only contains standard algorithms for the group operations and a sampling algorithm (cf.\ Section~\ref{sec:notation}). However, we show that the previous Theorem extends to this case, i.e., when only having a sampling algorithm. For the sake of readability, we will first treat the case of sampling \emph{uniformly at random} from $\nsgroup$ (Section~\ref{sec:uniformextension}), and will then show the general case with arbitrary (possibly unknown) sampling from $\nsgroup$ (Section~\ref{sec:generalextension}).

\subsection{Breaking SMP with Uniform Sampling}\label{sec:uniformextension}
It is well-known that if we have a sampling algorithm for $\nsgroup$ that samples \emph{uniformly at random}, we can obtain a set of generators by sampling polynomially (in the base-2 logarithm of the order of $\nsgroup$) many times from $\nsgroup$. If $k=\lceil\log_2(|\nsgroup|)\rceil$, Pak and Bratus~\cite{pak99sampling} show that $k+4$ samples are sufficient to generate the whole group with probability $> 3/4$. This result is an improvement over a result by Erd\"os and R\'enyi~\cite{erdos65sampling}. We recall it in the following Theorem:
\begin{theorem}[Probability of Finding a Generating Set with Uniform Sampling~\cite{pak99sampling}]\label{thm:unifsampl}
Let $\nsgroup$ be a finite abelian group of order $n$ where $k = \lceil \log_2(n) \rceil$. Then:
\[\prob{x_1,\ldots,x_{k+4}\fromunif\nsgroup}{\genby{x_1,\ldots,x_{k+4}}=\nsgroup}>\frac{3}{4}.\]
\end{theorem}
As an immediate corollary of this Theorem and Theorem~\ref{thm:membgen} we have the main result of this section.
\begin{theorem}[Quantum Attack on SMP with Uniform Sampling]\label{thm:final}
Let $(\group,\nsgroup)$ be the output of $\gen(\secparam)$ with $k=\lceil \log_2(|\nsgroup|) \rceil$, for some security parameter $\secparam$, such that the sampling algorithm in the description of $\nsgroup$ samples \emph{uniformly at random} from $\nsgroup$. Then, there exists a quantum algorithm which breaks the hardness of \SMP\ with probability at least $\frac{3}{4}\left(1-\nepsilon\right)^2$ in time $o(\poly(k+\log_2(1/\nepsilon)))$, and by sampling only $k+4$ times from $\nsgroup$.
\end{theorem}
We remark that the constant $\frac{3}{4}$ can be greatly improved by increasing the number of samples we take from $\nsgroup$, approximating $1$ very quickly.
In general, by performing $k+l$ random sampling, the success 
probability approximates $1$ exponentially fast in $l$.

\subsection{Breaking SMP with Arbitrary/Unknown Sampling}\label{sec:generalextension}
In this section, we show an extension of Theorem~\ref{thm:final} to the general case, where the description of $\nsgroup$ only contains a sampling algorithm with {\em unknown/arbitrary distribution} $\dist$. Since we do not make any restrictions on the sampling algorithm, we  cannot exclude seemingly exotic cases where parts  of $\nsgroup$ are hardly (or not at all) reached by the sampling algorithm. Consider the following example:

\begin{example}
Let $\secparam\geq 1$ be the security parameter. We define a family of groups by $\group_\secparam:=GF(2)^\secparam$ together with sampling distributions $\dist_\secparam$ on $\group_\secparam$ as through the probability mass function
\begin{equation}
\dist_\lambda(v_1,\ldots,v_\secparam):=
\left\{ \begin{array}{ll}
\frac{1}{2^{\secparam-1}}-\frac{1}{2^{\secparam\cdot (\secparam-1})}   & \mbox{, if } v_1=0 \\
\frac{1}{2^{\secparam\cdot (\secparam-1})} & \mbox{, otherwise.}
\end{array}\right.
\end{equation}
Here, $(v_1,\ldots,v_\secparam)$ denotes an arbitrary element from $GF(2)^\secparam$. Observe that the probability of sampling one vector $(v_1,\ldots,v_\secparam)$ with $v_1=1$ is $2^{-\secparam}$. However, at least one such sample is necessary for a generating set of the whole group. This shows that the probability of sampling a generating set for the whole group is negligible in $\secparam$.
\end{example}
As the examples illustrates, the best we can aim for (in general) is to find a generating set for a subgroup of $\nsgroup$ such that the probability that a random sample (with respect to $\dist$) does fall into this group is sufficiently large. This motivates the following definition:
 \begin{definition}[Covering Subgroup]
 Let a finite group $\nsgroup$ be given, together with a sampling distribution $\dist$. For a value $0\leq \cover\leq 1$, we say that a subgroup $\nsgroup^*\leq \nsgroup$ is a \emph{$\cover$-covering subgroup} of $\nsgroup$ with respect to $\dist$ if 
 \begin{equation}
\prob{x \fromdistr{\dist}\nsgroup}{x\in\nsgroup^*}\geq \cover.
 \end{equation}
 \end{definition}
 \begin{example}
Observe that the whole group $\nsgroup$ is trivially a $\cover$-covering subgroup. A less trivial example is the following. We order the elements $h\in\nsgroup$ in descending order according to their probabilities of being sampled, that is $h_1,h_2,\ldots$ with $\dist(h_i)\geq \dist(h_{i+1})$ for all $i$. Now, let $b$ denote the smallest index such that $\sum_{i=1}^b\dist(h_i)\geq \cover$. Then
 $\genby{h_1,\ldots,h_b}$ is for sure a $\cover$-covering subgroup. 
 \end{example}
 
Obviously, it follows directly from Theorem \ref{thm:membgen} that given generators of a $\cover$-covering subgroup, there exists a quantum attack on SMP with success probability at least $\cover\cdot \left(1-\nepsilon\right)^2$ in time $o(\poly(k+\log_2(1/\nepsilon)))$. Thus in the remainder of this section, we consider the task of finding, with probability $\geq \coversuccess$, a generating set for a $\cover$-covering subgroup (for  fixed, but arbitrary  values $\cover,\coversuccess$) if only a sampling algorithm $\Sample$ is given which samples according to an arbitrary (possibly unknown) distribution $\dist$. 
To this end, we prove the following new result on the probability of finding a $\cover$-covering subgroup (with generators) of a finite group with arbitrary/unknown sampling distribution and a given value $\cover$.

\begin{theorem}[Sampling a Generating Set for a $\cover$-covering Subgroup]\label{th:main_arbitrary}
Let $\nsgroup$ be a finite group, together with a sampling algorithm $\Sample$ that samples according to a (possibly unknown) distribution $\dist$, and let $k=\lceil\log_2(|\nsgroup|)\rceil$. Moreover, fix two values $0\leq \cover,\coversuccess\leq 1$ and set $\coversamples:=\loosebound$.

Let $x_1,\ldots,x_{\coversamples\cdot k+1}\in\nsgroup$ be $\coversamples\cdot k+1$ samples from $\nsgroup$  by invoking the sampling algorithm, i.e., $x_i\leftarrow \Sample$ for $i=1,\ldots,\coversamples\cdot k+1$.
Then with probability at least $\coversuccess$, the group $\nsgroup^*:=\genby{x_1,\ldots,x_{\coversamples\cdot k+1}}$  is a $\cover$-covering subgroup of $\nsgroup$.
\end{theorem}
Observe that like in the case of uniform sampling, a polynomial number of samples (almost linear in $k$) is sufficient. Interestingly, this number of samples is independent of the distribution.
 
For the sake of readability, we prove Theorem \ref{th:main_arbitrary} in two steps.
In the first step, we present an algorithm (Algorithm~\ref{alg:sampling_with_test}) that makes {\em at most} $\coversamples\cdot k+1$ samples and outputs a set $S\subset\nsgroup$. We prove that $S$ is a generating set for a $\cover$-covering subgroup with probability at least $\coversuccess$. 
The algorithm relies on the assumption of the existence of an efficient membership testing procedure. But in the second step we present a modification of the algorithm, Algorithm \ref{alg:no_test}, that works {\em without} the membership testing procedure and has at least the same success probability. In fact, Algorithm~\ref{alg:no_test} makes {\em exactly} $\coversamples\cdot k+1$ samples, hence proving Theorem \ref{th:main_arbitrary}.

\begin{algorithm}[t]
\caption{Sample generating set of a $\cover$-covering subgroup\label{alg:sampling_with_test}}
\begin{algorithmic}[1]
\REQUIRE A group $\nsgroup$ with sampling algorithm $\Sample$, an integer $k = \lceil \log_2{|\nsgroup|} \rceil$, a membership testing procedure that efficiently tests for any subset $S\subset\nsgroup$ and any $x\in \nsgroup$ whether $x\in\langle S\rangle$, two real values $0\leq \cover,\coversuccess\leq 1$.
\ENSURE A set $S$ of elements that generate a $\cover$-covering subgroup of $\nsgroup$ with probability at least $\coversuccess$.
\STATE \ 
\STATE $x\leftarrow\Sample$, $S \leftarrow \{x\}$\hfill\COMMENT{Initial candidate for a generating set}\label{algline:init_set}
\STATE $\coversamples:=\loosebound$ \hfill\COMMENT{Number of samples per round}
\STATE \ 
\FOR{$j=1,\ldots,k$}
\STATE $x_i\leftarrow\Sample,i=1,\ldots,\coversamples$\hfill\COMMENT{Sample $\coversamples$ elements from $\nsgroup$}
\IF{$x_i\in \langle S\rangle$ for all $i=1,\ldots,\coversamples$} 
\STATE Abort \textbf{for}-loop\hfill\COMMENT{Abort as all samples are already in $\langle S\rangle$} \label{algline:abort1}
\ELSE  
\STATE $S\leftarrow S\cup \{x_1,\ldots,x_\coversamples\}$\hfill\COMMENT{Extend candidate generating set}\label{algline:ext_set}
\ENDIF
\ENDFOR
\STATE\ 
\RETURN $S$
\end{algorithmic}
\end{algorithm}

\begin{algorithm}[t]
\caption{Sample generating set of a $\cover$-covering subgroup\label{alg:no_test}}
\begin{algorithmic}[1]
\REQUIRE A group $\nsgroup$ with sampling algorithm $\Sample$, an integer $k = \lceil \log_2{|\nsgroup|} \rceil$, and two real values $0\leq \cover,\coversuccess\leq 1$
\ENSURE A set $S$ of elements that generate a $\cover$-covering subgroup of $\nsgroup$ with probability $\geq\coversuccess$
\STATE \ 
\STATE $x\leftarrow\Sample$, $S \leftarrow \{x\}$\hfill\COMMENT{Initial candidate for a generating set}
\STATE $\coversamples:=\loosebound$ \hfill\COMMENT{Number of samples per round}
\STATE \ 
\FOR{$j=1,\ldots,k$}
\STATE $x_i\leftarrow\Sample,i=1,\ldots,\coversamples$\hfill\COMMENT{Sample $\coversamples$ elements from $\nsgroup$}
\STATE $S\leftarrow S\cup \{x_1,\ldots,x_\coversamples\}$\hfill\COMMENT{Extend candidate generating set}
\ENDFOR
\STATE\ 
\RETURN $S$
\end{algorithmic}
\end{algorithm}

We start with Algorithm~\ref{alg:sampling_with_test} and prove the following result:
\begin{theorem}[Correctness of Algorithm~\ref{alg:sampling_with_test}]
With a probability of at least $\coversuccess$, the output $S$ of Alg.~\ref{alg:sampling_with_test} is a generating set for a $\cover$-covering subgroup.
\end{theorem}
\begin{proof}
Let $S$ denote the output of Alg.~\ref{alg:sampling_with_test} and $\nsgroup^*:=\langle S\rangle$. There are two possibilities: (i) the algorithm aborted the \textbf{for}-loop for some value $j<k$ or (ii) the algorithm executed all $k$ \textbf{for}-loops. 

First, we consider case (i). At the same time, assume that $\nsgroup^*$ is \emph{not} a $\cover$-covering subgroup, that is $$\cover^*:=\prob{}{x\in\nsgroup^*\vert x \fromdistr{\dist}\nsgroup}< \cover$$ (this would be a failure of the algorithm). As the algorithm aborted the \textbf{for}-loops for some value $j<k$ by assumption, this can only happen if $x_i\in\langle S\rangle=:\nsgroup^*$ for all $\coversamples$ samples made in round $j$ although $\cover^*<\cover$. 
As the samples are made independently, the probability of this error event happening at a certain round is $\left(\cover^*\right)^N< \cover^\coversamples$; since there are at most $k-1$ independent rounds in case (i), the probability that an error occurs in any of them is at most 
$k\cdot \cover^\coversamples < 1-\coversuccess$ by definition of $\coversamples$. Hence, the probability that no error happens and the output is correct, i.e., is a generating set of a $\cover$-covering subgroup, is at least $1-(1-\coversuccess)=\coversuccess$. This concludes the first case.

Now, we consider case (ii), i.e., the algorithm has executed all $k$ \textbf{for}-loops. 
 For simplicity, we index the sets $S$ according to the round number. More precisely, let $S_0$ denote the initial candidate for the generating set (line \ref{algline:init_set}). Moreover, let $S_\ell$ denote the set  $S$ at the end of the while loop (after being extended - see line \ref{algline:ext_set}) and we define $\nsgroup_\ell:=\langle S_\ell\rangle$ for $\ell\geq 0$. Observe that $\nsgroup_\ell\subseteq\nsgroup$  for all $\ell$ by construction. The output of the algorithm is $S=S_k$. We make use of the following inequalities that we prove afterwards: 
\begin{equation}
\ord(\nsgroup_\ell)\geq 2^\ell\quad, \forall \ell\geq 0.\label{eq:claim}
\end{equation}
A consequence of \eqref{eq:claim} is that $
\ord(\nsgroup_k)\geq 2^k\geq \ord(\nsgroup)
$ which implies that $\nsgroup_k=\nsgroup$. Hence, $\nsgroup^*=\nsgroup_k=\nsgroup$ is the whole group and trivially a $\cover$-covering group for any value $0\leq \cover\leq 1$. 

It remains to prove the inequalities in \eqref{eq:claim}, i.e., $\ord(\nsgroup_\ell)\geq 2^\ell$ for all $0\leq \ell\leq k$. Observe that $\nsgroup_\ell$ is a proper subgroup of $\nsgroup_{\ell + 1}$ for every $\ell < k$. Thus, the number $\frac{\left| \nsgroup_{\ell+1}\right|}{\left| \nsgroup_\ell\right|}$ (which is an integer, by Lagrange's Theorem), must be strictly greater than $1$. Hence $\left| \nsgroup_{\ell+1}\right| \geq 2 \left| \nsgroup_\ell\right|$, and  this proves \eqref{eq:claim} since $\left| \nsgroup_0 \right|=1$. \qed

\end{proof}

Observe that Alg.~\ref{alg:sampling_with_test} runs at most $k$ for-loops and uses the membership test procedure only for deciding if the algorithm can be stopped earlier. Hence, we consider a variant, namely Alg.~\ref{alg:no_test}, which simply drops this test and always runs all $k$ loops. That is, the only difference between Algorithms \ref{alg:sampling_with_test} and \ref{alg:no_test}, respectively, is that the latter may run longer (but still at most $k$ loops) and outputs a superset $S'$ of the output $S$ of Alg.~\ref{alg:sampling_with_test}. Of course, if $S$ is a generating set for a $\cover$-covering subgroup, then this is certainly true for $S'$ as well. This shows that Alg.~\ref{alg:no_test} ``inherits'' the success probability of Alg.~\ref{alg:sampling_with_test}:
\begin{corollary}\label{cor:no_test}[Correctness of Algorithm~\ref{alg:no_test}]
With a probability of at least $\coversuccess$, the output $S$ of Algorithm \ref{alg:no_test} is a generating set for a $\cover$-covering subgroup.
\end{corollary}
Observe that Alg.~\ref{alg:no_test} simply outputs $\coversamples \cdot k + 1$ samples. Hence, the proof of Theorem \ref{th:main_arbitrary} is a direct consequence of Corollary~\ref{cor:no_test}. The remainder of this section is straightforward. Given  a generating set $S$ of a $\cover$-covering subgroup, we can apply Theorem~\ref{thm:membgen} in order to break the $\SMP$ for $(\group,\nsgroup)$.

\begin{theorem}[Quantum Attack on SMP with Arbitrary Sampling]\label{thm:final2}
Let $(\group,\nsgroup)$ be the output of $\gen(\secparam)$ with $k=\lceil \log_2(|\nsgroup|) \rceil$, for some security parameter $\secparam$. We denote the distribution of the sampling algorithm contained in the description of $\nsgroup$ by $\dist$. Let $0\leq \nepsilon^*\leq 1$ be an arbitrary fixed positive value.  Then, there exists a value $\coversamples = \coversamples(k,\nepsilon^*)$ (which only grows at most logarithmically in $k$) and a quantum algorithm which breaks the hardness of \SMP\ with probability at least $\left(1-\nepsilon^*\right)\left(1-\nepsilon\right)^2$ in time $o(\poly(k+\log_2(1/\nepsilon)))$, and by sampling only $\coversamples\cdot k+1$ times from $\nsgroup$ (where $\nepsilon$ is the error probability of Theorem~\ref{thm:ordgen}).

In particular, we can construct a quantum algorithm that breaks \SMP\ with probability at least $\frac{3}{4}(1-\nepsilon)^2$ in time $o(\poly(k+\log_2(1/\nepsilon)))$ while only sampling $7k\cdot \left(2+\lceil\log(k)\rceil\right)+1$ times from $\nsgroup$.
\end{theorem}
\begin{proof}
In principle, the attacker $\adv$ is the same as described in Theorems~\ref{thm:membgen} and~\ref{thm:final}, the only difference being the approach for finding an appropriate generating set. Given the value $\nepsilon^*$, the attacker chooses two positive values $\cover,\coversuccess$ such that $\cover\cdot\coversuccess\geq (1- \nepsilon^*)$, for example $\cover=\coversuccess=\sqrt{1- \nepsilon^*}$. Then, the attacker makes $\coversamples\cdot k+1$ samples as explained in Theorem~\ref{th:main_arbitrary}. Let $\nsgroup^*$ denote the subgroup of $\nsgroup$ that is generated by these $\coversamples\cdot k+1$ samples. Due to Corollary~\ref{cor:no_test}, we know that $\nsgroup^*$ is a $\cover$-covering subgroup of $\nsgroup$ with probability $\coversuccess$. From this point on, the attack continues as specified in Theorem~\ref{thm:membgen}, while using the $\coversamples\cdot k+1$ samples as generators, i.e.,
we let $z$ denote the challenge in the \SMP\ game (Def.~\ref{def:SMP}), so $z\in\group\setminus\nsgroup$ if $b=1$, and $z\in\nsgroup$ otherwise. If $b=1$ (which happens with probability $\frac{1}{2}$), we know that $z\not\in\nsgroup^*$ and the attacker $\adv$ will recognize this with probability $\geq (1-\nepsilon)^2$ (as in the proof of Theorem~\ref{thm:membgen}).  If $b=0$ (which also happens with probability $\frac{1}{2}$), several sub-cases do exist (depending on whether $\nsgroup^*$ is $\cover$-covering and  whether $z\in\nsgroup^*$). In case that both properties are true (which happens with probability $\geq \coversuccess\cdot\cover)$, the attacker recognizes that $z\in\nsgroup^*$ again with probability $\geq (1-\nepsilon)^2$. As the success probabilities in the other sub-cases are at least zero, it follows that
\begin{equation*}
\prob{}{\exp{\smp}{\adv,\gen}(\secparam)=1}\geq \frac{(1-\nepsilon)^2+\cover\coversuccess (1-\nepsilon)^2}{2}
\geq \cover\coversuccess(1-\nepsilon)^2
\geq  (1- \nepsilon^*)(1-\nepsilon)^2
\end{equation*}
%
which concludes the proof of the first part of the Theorem. For the second part, we see that when choosing $\nepsilon^*=\frac{1}{4}$ and $\cover=\coversuccess=\frac{1}{2}\sqrt{3}$, the above attacker $\adv$ has a success probability of at least $\frac{3}{4}(1-\nepsilon)^2$ by sampling only $\coversamples \cdot k + 1$ times from $\nsgroup$ where $\coversamples = \loosebound \leq 7 \left( \lceil\log(k)\rceil+2 \right)$.
\qed
\end{proof}
Finally, Theorems~\ref{thm:final2} and~\ref{thm:cpacondition} together immediately imply our main result: the general impossibility of group homomorphic encryption in the quantum world, if the plaintext and ciphertext groups are abelian. 
\begin{theorem}[Impossibility of Group Homomorphic Encryption in the Quantum World]\label{thm:finalll}
Let $\encsch=(\keygen,\enc,\dec)$ be an \indcpa\ secure group homomorphic encryption scheme with abelian plaintext and ciphertext groups. Then, there exists a quantum PPT algorithm that breaks the security of $\encsch$ with non-negligible probability.
\end{theorem}
%


\section{Discussion}\label{sec:discussion}
In this section, we provide an informal discussion about the applicability of our quantum attack to non-group homomorphic encryption schemes and elaborate on fully homomorphic encryption (FHE).
In abstract terms, existing FHE schemes are standard public-key encryption schemes $\encsch=(\keygen,\enc,\dec)$ with the following extras~\cite{gentry09fhe}:
\begin{itemize}
\item the plaintext space $\ptexts$ and ciphertext space $\ctexts$ are rings,

\item there is an algorithm $\eval$ that takes as input a public key $\pk$, a circuit $C$, a tuple $(c_1,\ldots,c_t)$ of ciphertexts (one for every input node of $C$), and outputs another ciphertext $c$, and

\item for all outputs $(\pk,\sk)$ by $\keygen(\secparam)$, all polynomials $p(\secparam)$ in $\secparam$, all $t\leq\poly(\secparam)$, all plaintexts $m_1,\ldots,m_t\in\ptexts$ corresponding to fresh encryptions $c_i\fromdist\enc_\pk(m_i)$, $i=1\ldots t$, and all $t$-input circuits $C$ of depth $\leq p(\secparam)$, we have the following \emph{correctness} condition:
\begin{equation}\label{eqn:correctness}
\dec_\sk(\eval_\pk(C,c_1,\ldots,c_t))=C(m_1,\ldots,m_t).
\end{equation}
\end{itemize}
Homomorphic encryption schemes for which the polynomial depth $p(\secparam)$ of the circuits $C$ is bounded a priori (i.e., fixed in the public key $\pk$) are called \emph{leveled} FHE. For very small polynomials $p(\secparam)$, we say that the scheme is \emph{somewhat} homomorphic. At a first glance, there a two main differences to the notion of \emph{group} homomorphic encryption (see Fig.~\ref{fig:grphom} for a pictorial explanation):
\begin{enumerate}
\item The set of all (fresh) encryptions $\allencs=\{\enc_\pk(m;r)\mid m\in\ptexts, r\in\Rnd\}$ is only a \emph{subset} (and not necessarily a subgroup) of the ring $\ctexts$.

\item The decryption is not necessarily a group homomorphism as it is only guaranteed to run correctly with circuits that are polynomially bounded in depth; this polynomial bound can be dynamically chosen in the ``pure'' FHE case, while it is fixed in the public key for leveled FHE and somewhat homomorphic schemes. But if the decryption is group homomorphic, it particularly must run correctly (at least theoretically) on all unbounded circuits consisting only of group-operation gates.
\end{enumerate}
\begin{figure}[h]
\begin{center}
\framebox[.8\textwidth]{\includegraphics[width=.75\linewidth]{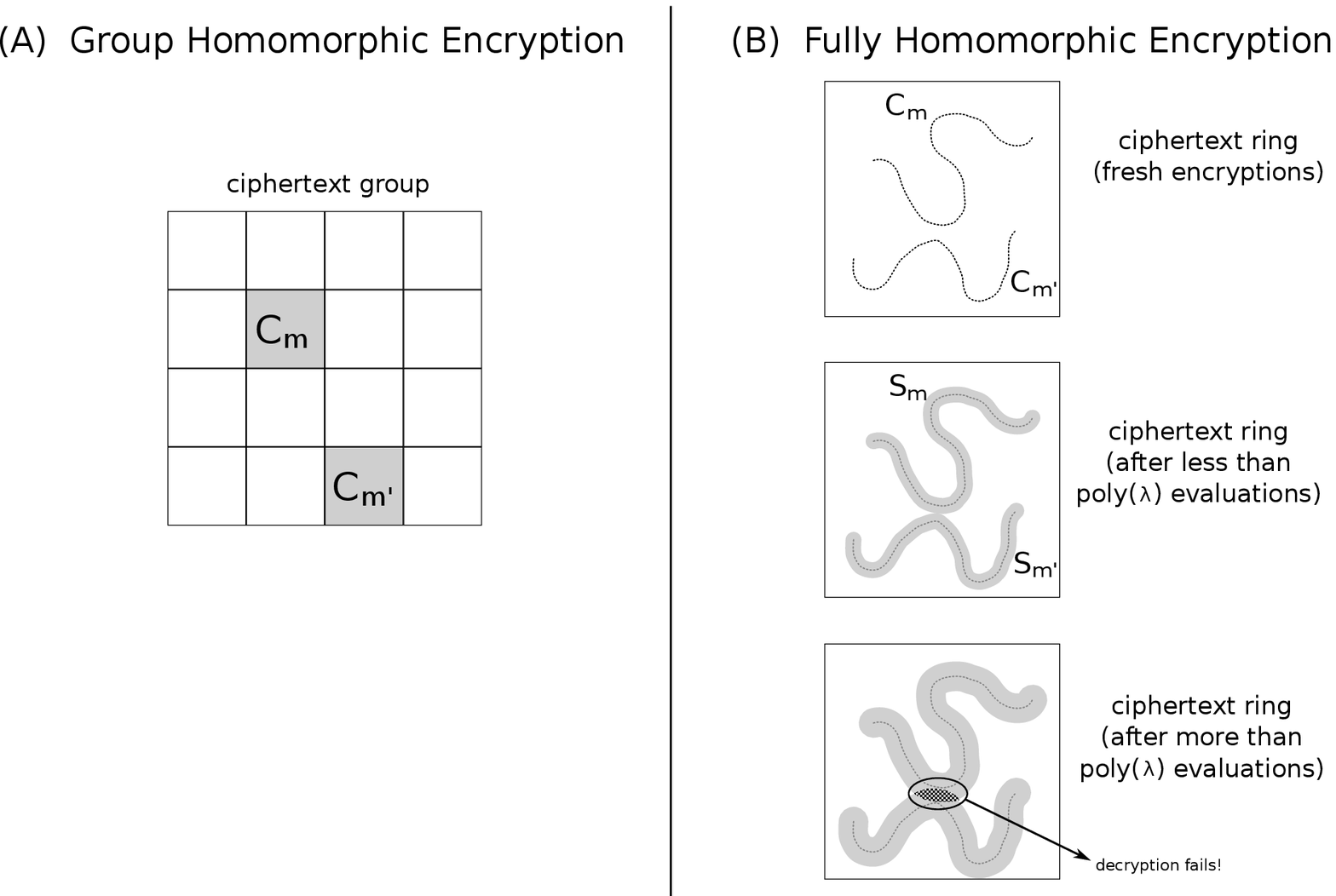}}
\end{center}
\vspace{-.5cm}
\caption{Differences between group homomorphic encryption and FHE: (A) shows that each $\cipher{m}$ is a coset of $\cipher{1}$ in $\cipher{}$ (Fact~\ref{lemma:c0}), while the decryption is a group homomorphism; (B) shows first that $\cipher{m}$ and $\cipher{m'}$ are subsets and not necessarily cosets in $\cipher{}$, second that the decryption runs correctly on $\poly(\secparam)$ evaluations of ciphertexts, and third that the decryption might fail if exponentially many evaluations have been performed, meaning that the decryption is not necessarily group homomorphic.}\label{fig:grphom}
\end{figure}
If the decryption is not a group homomorphism, the set of fresh encryptions of the neutral element in $\ptexts$ is not necessarily a group, but only a subset of $\ctexts$. However, the quantum order-finding algorithm of Theorem~\ref{thm:ordgen} only works on (solvable) \emph{groups}. This immediately gives us the first important observation:
\begin{observation}
Our quantum attack from Section~\ref{sec:quantum} on group homomorphic encryption schemes is \emph{not} immediately applicable to more general homomorphic encryption schemes, such as somewhat and (leveled) FHE schemes.
\end{observation}
A sufficient condition that we need a homomorphic scheme to have for our quantum attack to work is the following:

\medskip
\noindent{\bf Sufficient Condition (Quantum Attack).} For any output $(\pk,\sk)$ by $\keygen(\secparam)$, there exist two plaintexts $m,m'\in\ptexts$ and a subgroup $\group$ of $\ctexts$ such that 
\begin{enumerate}
\item there exists an efficient PPT algorithm which outputs a generating set for $\group$ of size at most $\poly(\secparam)$,


\item the probability $\prob{c\fromdist\enc_\pk(m)}{c\in\group}$ is non-negligible in $\secparam$, and

\item the probability $\prob{c'\fromdist\enc_\pk(m')}{c' \notin \group}$ is non-negligible in $\secparam$.
\end{enumerate}
In the setting of group homomorphic encryption schemes, the plaintext $m$ would be the neutral element $1$, while $m'\neq 1$ can be any other plaintext. The group $\group$ satisfying the above conditions would be a $\delta$-covering subgroup of the group $\cipher{1}$ of all (fresh) encryptions of $1$, for a sufficiently small $\delta$.
For more general homomorphic encryption schemes, such as somewhat or (leveled) FHE schemes, the situation looks more like in Fig.~\ref{fig:fhe}.
\begin{figure}[h]
\begin{center}
\framebox[.8\textwidth]{\includegraphics[width=.75\linewidth]{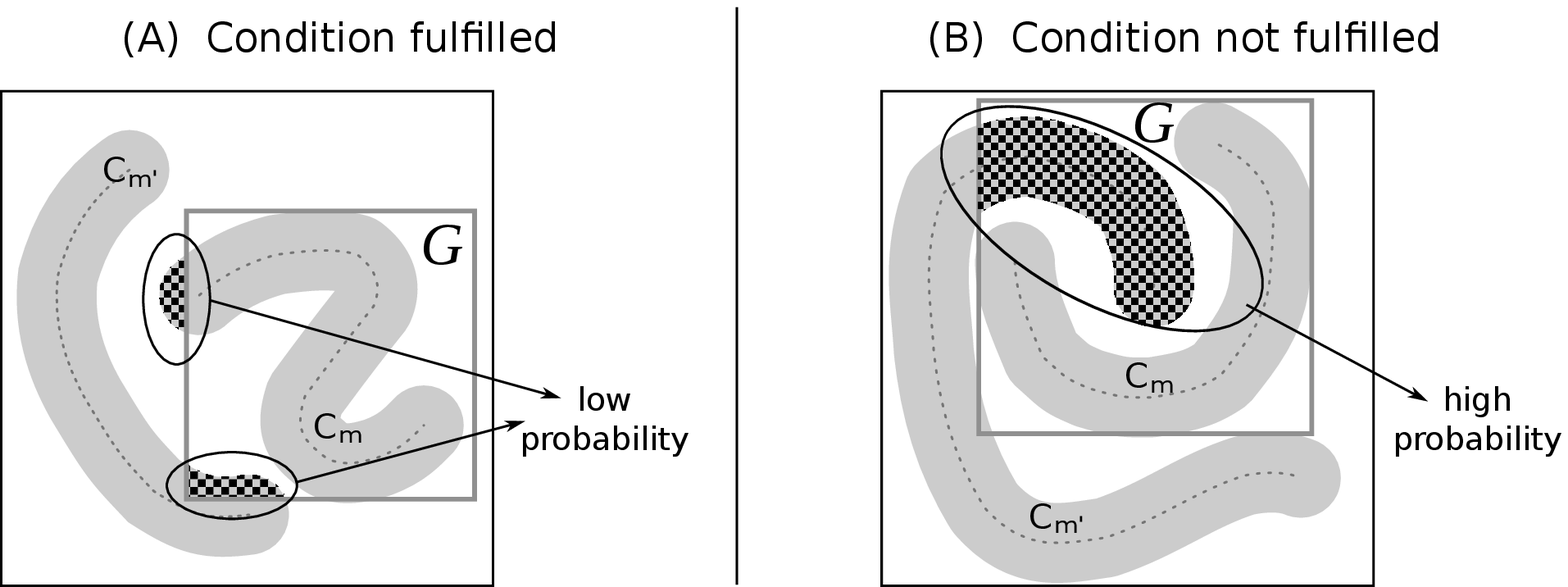}}
\end{center}
\vspace{-.5cm}
\caption{Our condition in the FHE case: (A) shows pictorially when the condition is fulfilled; (B) shows the case when item 3 of the condition is not met and $G$ intersects with a large part of encryptions of $m'$.}\label{fig:fhe}
\end{figure}

\noindent The important observation here is, that as long as only polynomially many evaluations of the ciphertexts have been performed, the decryption still runs correctly (cf.\ correctness condition in Equation~\eqref{eqn:correctness}). But for any scheme to be \indcpa\ secure, the set of encryptions of a given message $m$ must be exponentially large, so in particular, a group $\group$ that fulfills condition 2 is required to be exponentially large. Hence, the decryption is not guaranteed to run correctly on $\group$ and might fail. More precisely, condition 3 for our attack to work will most likely be unsatisfied. However, proving or disproving that any of the existing somewhat or (leveled) FHE schemes satisfies our sufficient condition is a highly non-trivial task (due to the very general and abstract nature of the requirement) and lies outside the scope of this work. We leave it as interesting future work. Interestingly enough, since most of the existing FHE schemes base their security on supposedly quantum-resistant hardness assumptions (such as LWE), spotting a scheme that is susceptible to our quantum attack will effectively break the underlying hardness assumption and thereby disprove its quantum-resistance.



\subsubsection*{Acknowledgements.}
We would like to thank Richard Lindner and Pooya Farshim for helpful discussions. We are also grateful for the constructive comments by the anonymous reviewers.


\bibliographystyle{splncs03}
\bibliography{main}

\appendix
\section{Proof of Theorem~\ref{thm:cpacondition}}\label{apx:proofofcond}
We prove the theorem by contradiction and show that if we have a PPT algorithm $\adv$ that breaks the hardness of $\SMP$ with non-negligible advantage $\nnegl(\secparam)$, we can construct (in PPT) an algorithm $\bdv$ that breaks the \indcpa\ security with non-negligible advantage $\nnegl(\secparam)$. To this end, we fix an $\SMP$-adversary $\adv$ and construct an \indcpa -adversary $\bdv=(\bdv_1,\bdv_2)$.

We start by letting $\bdv_1$ choose $m_0=1\in\ptexts$ and a random message $m_1\output\ptexts$ with $m\neq 1$. Next, 
$\bdv_1$ sends the two messages $m_0,m_1$ to the \indcpa -challenger. The challenger chooses a random bit $b\in\{0,1\}$ and returns the ciphertext $c\output\enc_\pk(m_b)$. Then, $\bdv_2$ simply relays the ciphertext $c$ to the $\SMP$-adversary $\adv$ who will output a bit $d\in\{0,1\}$,
which in turn is forwarded by $\bdv_2$ to the \indcpa -challenger.

It remains to be shown that $d=b$ with a non-negligible advantage, i.e., that
\[\left|\prob{}{\indexp{\cpa}{\bdv,\keygen}(\secparam)=1}-\frac{1}{2}\right|\text{ is non-negligible}.\]
By the assumption on $\adv$, we know that $\adv$'s advantage is non-negligible, namely $\nnegl(\secparam)$. Moreover, the ciphertext $c$ is formatted as in $\SMP$ so $\adv$ behaves as in the $\SMP$-game (it is either a fresh encryption of $1$ or of a random message different from 1), meaning that $d=b$ with $\adv$'s advantage $\nnegl(\secparam)$, i.e.,
\[\left|\prob{}{\indexp{\cpa}{\bdv,\keygen}(\secparam)=1}-\frac{1}{2}\right|=\left|\prob{}{\exp{\smp}{\adv,\keygen}(\secparam)=1}-\frac{1}{2}\right|=\nnegl(\secparam).\]
This concludes the proof of the Theorem.\qed

\section{Example: Hardness of SMP does NOT imply IND-CPA Security}\label{apx:example}
We construct a \GHE\ scheme that is \emph{not} \indcpa\ secure, but whose corresponding \SMP\ is hard. In a nutshell, the idea is to start with a IND-CPA secure scheme but to change the encryption process as follows. For a fixed message $m^*\neq 1$ the encryption process becomes deterministic for a significant probability (e.g., 1/2). An IND-CPA attacker can misuse this to easily distinguish encryptions of $m^*$ from other ciphertexts. However, if the plaintext space is sufficiently large, the probability that the SMP-sampling algorithm chooses $m^*$ is negligible, leaving the SMP still hard.

More precisely, let $\encsch=(\keygen,\enc,\dec)$ be an \indcpa\ secure \GHE\ scheme with a plaintext group $\ptexts$ that is exponentially large in the security parameter such that the sampling algorithm, contained in the description of $\ptexts$, samples according to the uniform distribution---for instance, this property is satisfied by the ElGamal cryptosystem~\cite{elgamal85elgamal}. By Theorem~\ref{thm:cpacondition}, we know that the \SMP\ relative to $\keygen$ is hard. Now, the idea is to slightly modify $\encsch$ such that the corresponding \SMP\ remains hard but the \indcpa\ security can be easily broken. Therefore, we fix a public value $r^*\in\Rnd$, a public message $m^*\in\ptexts\setminus\{1\}$, and construct a scheme $\encsch^*$ which is exactly the \emph{same} as $\encsch$, except for the encryption algorithm. We denote the encryption algorithm of $\encsch^*$ by $\enc^*$ and define it as follows:

\noindent{\bf Encryption.} $\enc^*$ takes the public key $\pk$, a message $m$, and a random value $r\in\Rnd$ as input. Furthermore, it uniformly samples a random bit $b^*\in \{0,1\}$. The output is defined as follows:
\[\enc^*_\pk(m;r):=
\left\{ \begin{array}{ll}
\enc_\pk(m;r^*)   & \mbox{, if } m=m^*\ \mathrm{and}\ b^*=0 \\
\enc_\pk(m;r)   & \mbox{, if } m=m^*\ \mathrm{and}\ b^*=1 \\
\enc_\pk(m;r) & \mbox{, otherwise.}
\end{array}\right.\]
Recall that $r^*\in\Rnd$ and $m^*\in\ptexts$ are fixed and public values corresponding to $\encsch^*$.

Our new scheme $\encsch^*$ certainly is \emph{not} \indcpa\ secure: Assume an adversary chooses two messages $m_0,m_1\in\ptexts$ where $m_0=m^*$. Upon the retrieval of an encryption $c$ of either of the two messages, the adversary checks whether $c=\enc_\pk(m;r^*)$. If so, she knows that $m_0$ was encrypted. Otherwise she assumes that $c$ is an encryption of message $m_1$. Her advantage is $1/4$.

On the other hand, we see that the \SMP\ corresponding to $\encsch^*$ is still hard: Recall that in the \SMP\ game, the challenger flips a coin $b\in\{0,1\}$. If $b=1$, the challenger samples a randomly chosen message $m\fromunif\ptexts$ with $m\neq 1$ (recall that sampling from $\ptexts$ is done according to the uniform distribution) and sends $c=\enc^*_\pk(m)$ to an \SMP -adversary. If $b=0$, the challenger simply sends $c=\enc^*_\pk(1)$ to the adversary. It is obvious that this \SMP\ instance (using $\enc^*$) behaves exactly in the same way as our orginial \SMP\ game (with $\enc$) corresponding to $\encsch$ if $b=0$. But also if $b=1$, it is clear that the advantage of an adversary in the \SMP\ with $\enc^*$ is negligibly close to the advantage of an adversary in the \SMP\ with $\enc$. This is due to the fact that the plaintext space is exponentially large in the security parameter and the particular message $m^*$ will only be chosen with a negligible probability. Therefore, the two games \SMP\ with $\enc^*$ and \SMP\ with $\enc$ are computationally indistinguishable, and so our \SMP\ corresponding to $\encsch^*$ is hard.

\end{document}